\documentclass[acmsmall,10pt]{acmart}
\acmJournal{TOPLAS}
\acmYear{2021}
\acmISBN{} \acmDOI{} \startPage{1}

\setcopyright{none}

\usepackage{booktabs}   \usepackage{subcaption} 

\usepackage{amsmath}
 \usepackage{amssymb}
\usepackage{bm}
\usepackage{mathtools}
\usepackage[nocomma]{optidef}
\usepackage{algorithm}
\usepackage{enumitem}
\usepackage{algorithmicx}
\usepackage{algpseudocode}
\usepackage[symbol]{footmisc}
\usepackage{xcolor}
\usepackage[frozencache]{minted}

\usepackage{enumitem}
\newlist{steps}{enumerate}{1}
\setlist[steps, 1]{label = \textbf{Step} \arabic*:}
\newcommand{\oomit}[1]{}

\DeclarePairedDelimiter\norm{\lVert}{\rVert}

\newtheorem{lemma}{Lemma}

\newtheorem{theorem}{Theorem}
\newtheorem{remark}{Remark}
\newtheorem{definition}{Definition}
\newtheorem{example}{Example}

\begin{document}

\title[Synthesizing Invariant Clusters for Polynomial Programs by SDP]{Synthesizing Invariant Clusters for Polynomial Programs by Semidefinite Programming}

\author{Qiuye Wang}
\orcid{0000-0001-5138-3273}
\affiliation{
    \department{State Key Lab. of Computer Science}
    \institution{Institute of Software, Chinese Academy of Sciences}
    \institution{University of CAS}
    \city{Beijing}
    \country{China}
}
\email{wangqye@ios.ac.cn}

\author{Bai Xue}
\affiliation{
    \department{State Key Lab. of Computer Science}
    \institution{Institute of Software, Chinese Academy of Sciences}
    \institution{University of CAS}
    \city{Beijing}
    \country{China}
}
\email{xuebai@ios.ac.cn}

\author{Naijun Zhan}
\affiliation{
    \department{State Key Lab. of Computer Science}
    \institution{Institute of Software, Chinese Academy of Sciences}
    \institution{University of CAS}
    \city{Beijing}
    \country{China}
}
\email{znj@ios.ac.cn}

\author{Lihong Zhi}
\affiliation{ 
    \department{Mathematics Mechanization Research Center}
    \institution{Institute of System Science, Academy of Mathematics and System Sciences, Academia Sinica}
    \city{Beijing}
    \country{China}
}
\email{lzhi@mmrc.iss.ac.cn}

\author{Zhi-Hong Yang}
\affiliation{
    \department{Department of Mathematics}
    \institution{NC State University}
    \state{NC}
    \country{USA}
}
\additionalaffiliation{
    \department{Department of Computer Science}
    \institution{Duke University}
    \state{NC}
    \country{USA}
}
\email{zhyang@amss.ac.edu}

\begin{abstract}

    In this paper, we present a novel approach 
    to synthesize \emph{invariant clusters} for polynomial programs. 
    An invariant cluster is a set of program invariants that share a common structure, 
    which could, for example, be used to save the needs for repeatedly synthesizing new invariants 
    when the specifications and programs are evolving. 
    To that end, we search for sets of parameters $R_k$ 
    w.r.t. a parameterized multivariate polynomial $I(\bm{a}, \bm{x})$ (i.e. a \emph{template}) 
    such that $I(\bm{a}, \bm{x}) \leq 0$ is a valid program invariant for all $\bm{a} \in R_k$. 
    Instead of using time-consuming symbolic routines such as quantifier eliminations, 
    we show that such sets of parameters can be synthesized 
    using a hierarchy of semidefinite programming (SDP). 
    Moreover, we show that, under some standard non-degenerate assumptions, 
    \emph{almost all} possible valid parameters can be included in the synthesized sets. 
    Such kind of completeness result has previously only been provided by symbolic approaches. 
    Further extensions such as 
    using semialgebraic and general algebraic templates (instead of polynomial ones) 
    and allowing non-polynomial continuous functions in programs are also discussed.  
\end{abstract}



\maketitle

\section{Introduction}
\label{sec:introduction}

The dominant approach to program verification 
is \emph{Floyd-Hoare-Naur's inductive assertion method} \cite{floyd67,hoare69,naur66}, 
which is based on Hoare Logic \cite{hoare69}. 
The hardest parts thereof are \emph{invariant generation} and \emph{termination analysis}. 
Powerful as this approach may be in theory, 
it has only limited success in the early stage \cite{wegbreit74, gw75, km76, karr76} 
as in practice it depends on what kind of invariants can be generated, 
and back then automatic generating only gives some special linear invariants. 

In the past two decades, safety-critical systems have become important parts of our life, 
a thorough validation and verification is necessary 
to enhance the quality of software used in these systems, 
and in particular, to fulfill the quality criteria mandated by relevant standards. 
This motivates many theories and computation approaches to synthesize better invariants, 
which have been successfully applied in the verification of programs and hybrid systems 
\cite{rck04a, ijcar16, Garg16, kincaid2017nonlinear, joel2018polynomial, Sankaranarayanan04, Kapur06}. 

Sometimes, it would often be beneficial to synthesize not only one valid invariant 
as the witness of system safety, 
but also a set of valid invariants with similar structures 
(that is, an \emph{invariant cluster} \cite{kong2017clusters} 
\footnote{
    The concept of invariant clusters is presented in \cite{kong2017clusters} 
    for hybrid systems (with continuous dynamics). 
    The idea of synthesizing a group of different program invariants, 
    has touched by many existing works, e.g. \cite{kapur2007generating, joel2018polynomial}. 
}
). 
For example, when the specifications and programs are evolving, 
the previous invariant may become inadequate 
to assure the validity of a slightly evolved version of a Hoare triple. 
The obvious way is to synthesize new invariants for every updates, 
which creates a heavy burden for updating. 
However, if an invariant cluster is available, 
one can simply try to draw another invariant from that cluster 
to fit the updated specifications and programs, 
which is generally much easier than doing the synthesizing procedure all over again. 
This idea of invariant clusters could be particular useful 
when one needs to analyse evolving programs 
where the \emph{reuse} of previous results is of central concerns \cite{oopsla2020termination}. 

Given an invariant template $I(\bm{a}, \bm{x})$, 
the condition for $I(\bm{a}, \bm{x}) \leq 0$ to be an invariant 
can be written as a quantified first order logic formula. 
If the template $I(\bm{a}, \bm{x})$, the program and specifications are all polynomial, 
Synthesizing an (in fact, the largest possible) invariant cluster can then be done by 
eliminating the quantifiers of this formula \cite{Kapur06}. 
The result will be a quantifier-free formula on the parameters $\bm{a}$, 
representing the constraints needed to be satisfied 
if $I(\bm{a}, \bm{x})$ are to be an invariant.  
However, quantifier elimination procedures are extremely time-consuming 
(with up to double-exponential time complexity), 
which means the above approach is only applicable in some tiny examples. 

In this paper, we present a novel approach to synthesize invariant clusters for polynomial programs. 
The key observation here is that we do not always need the \emph{largest possible} invariant cluster 
(or equivalently, the exact valid set $R_I$ containing \emph{all} valid parameters). 
Rather, sometimes it is enough to have underapproximations of $R_I$. 
From there on, we show that some underapproximations $R_k = \{ \bm{a} \mid p_i(\bm{a}) \leq 0 \}$ 
of the valid set $R_I$ can be synthesized 
by solving a hierarchy of semidefinite programmings (SDP) \cite{Lasserre10, Lasserre15}, 
where $\{ p_i \}$ is a series of polynomials with increasing degrees. 
We also show that, under some standard non-degenerate assumptions, 
the series of underapproximations $\{ R_k \}$ converges to the valid set $R_I$ 
w.r.t. Lebesgue measure as the relaxation degree increases. 
Therefore, by choosing a suitable relaxation degree, 
one can obtain an underapproximation of the valid set $R_I$ with desired precision. 

On top of the above algorithm that builds underapproximations by solving a SDP hierarchy, 
we also provide some additional techniques that increase its practical performance, 
such as adaptively partitioning of the parameter space 
and using alternative non-trivial polynomial basis. 
The issue of possible unsoundness caused by numerical errors in SDP is also discussed 
and several methods have been presented to ease its effects. 
Experiment results show that our approach is capable of 
synthesize invariant clusters that no previous methods could provide for some programs, 
and its performance is comparable with some state-of-the-art invariant synthesis methods. 

The main contributions of this paper are given as follows: 
\begin{enumerate}
    \item We presented the idea of invariant clusters in details, 
          which could be useful in analysing evolving programs. 
    \item We presented an algorithm to synthesize invariant clusters 
          based on Lasserre's SDP techniques presented in \cite{Lasserre15}. 
    \item We provided additional techniques that 
          increase the practical performance of this algorithm. 
\end{enumerate}

The rest of this paper is organized as follows: 
Section~\ref{sec:overview} gives an overview of our approach. 
In Section~\ref{sec:preliminaries}, 
we introduce some basic notions and tools that will be used later. 
Section~\ref{sec:underapproximation} is devoted to 
showing that the problem of synthesizing invariant clusters can be solved
using the techniques in \cite{Lasserre15}, which underapproximates a certain kind of sets. 
Section~\ref{sec:techniques} includes additional techniques 
that increase the practical performance of the underapproximation procedure. 
Some illustrative experiments are presented in Section~\ref{sec:experiments}, 
Section~\ref{sec:extension} discusses how to extend our approach 
to deal with more general programs with more expressive templates. 
Finally, we conclude this paper and discuss future work in Section~\ref{sec:conclusion}. 
 \section{Overview}\label{sec:overview}
In this section, we use a practical example to show that how our approach works.
\subsection{An Example}
The example program in Code~\ref{code:overview} is taken from \cite{DXZ13}. 
The variables $x$ and $y$ are assumed to be real-valued. 
One needs to verify the following safety property: 
if the initial state satisfies the precondition $x^2 + y^2 \leq 1$, 
then the postcondition $x^2 - 2 y^2 \leq 4$ holds whenever the loop terminates, 

\begin{listing}[!ht]
\begin{minted}[mathescape]{c}
    // Precondition: {$x^2 + y^2 \leq 1$} 
    while (x * x + y * y <= 3) {
        x = x * x + y - 1;
        y = x * y + y + 1;
    }
    // Postcondition: {$x^2 - 2 y^2 \leq 4$}
\end{minted}
\caption{The Overview Example}
\label{code:overview}
\end{listing}

It should be noted that even though this example seems simple enough, 
many existing works failed to synthesize even one suitable invariant for it 
\cite{joel2018polynomial, CSS03, kincaid2017nonlinear}, 
let alone an invariant cluster. 
On the other hand, directly applying symbolic constraint solving methods such as SMT solving 
is impractical for this nonlinear program 
due to the extremely high computational complexity (up to double exponential\cite{DH88}). 

\subsection{Verification by Invariants}
The safety property in Code~\ref{code:overview} can be proved 
by finding a suitable invariant $\phi(x, y)$ such that 
\begin{itemize}
    \item the precondition implies $\phi$, 
    \item the truth value of $\phi$ stay unchanged during execution, 
    \item when loop terminates, $\phi$ implies the postcondition. 
\end{itemize}

We employ the idea of \emph{template-based invariant synthesis} 
and take a template $\mathit{Inv}(a, b, x, y) = x^2 + a y^2 + b$ 
(as a simple heuristic, 
the template is taken to include all monomials in the postcondition). 
Now, it can be subsequently proved that 
$\mathit{Inv}(a, b, x, y)$ is an invariant if 
the following first-order formulas are satisfied 
(for more details, see Definition~\ref{def:inv}):
\begin{equation*}
\begin{aligned}
    C_1 := \forall x \ldotp \forall y \ldotp & (x^2 + y^2 - 1 \leq 0) \implies (x^2 + a y^2 + b \leq 0) \, ;  \\
    C_2 := \forall x \ldotp \forall y \ldotp &
((x^2 + y^2 - 3 \leq 0) \land (x^2 + a y^2 + b \leq 0)) \implies \\ 
                                      & (a x^4 y^2 + 2 a x^2 y^3 + a y^4 + x^4 + 2 a x^2 y + 
                                        2 x^2 y + 2 a y^2 - 2 x^2 + y^2 - 2 y + a + b + 1) \, ; \\
C_3 := \forall x \ldotp \forall y \ldotp & ((x^2 + y^2 \geq 3) \land (x^2 + a y^2 + c \leq 0)) \implies 
                                       (x^2 + y^2 - 4 \leq 0) \, .  
\end{aligned}
\end{equation*}

The obvious way of finding the explicit conditions on 
the parameters $a$ and $b$ is to eliminate the universal quantifiers in the above formulas. 
After that, one obtains an exact representation of the valid set $R_I$ 
(the set of all valid parameters). 
Alternatively, if only one invariant is needed, 
one may also employ SMT solvers to find a satisfiable assignment of the above formulas. 
However, due to the nonlinear nature and complex forms of these formulas, 
neither of the two methods can give an answer in a reasonable time 
(\texttt{Redlog}\cite{dolzmann1997} and \texttt{Z3}\cite{Z3} 
both failed to provide results after $24$ hours of running). 

\subsection{Underapproximation of $R_I$ using SDP}
Our main idea is to find \emph{underapproximations} of the formulas $C_1$, $C_2$ and $C_3$.
More specifically, we try to find some polynomials $p_1$, $p_2$, $p_3$ such that: 
\begin{equation*}
\begin{aligned}
    p_1(a, b) \leq 0 & \implies C_1 \, , \\
    p_2(a, b) \leq 0 & \implies C_2 \, ,  \\
    p_3(a, b) \leq 0 & \implies C_3 \, .
\end{aligned}
\end{equation*}

In other words, $R = \{ (a, b) \mid p_i(a, b) \leq 0, \, i = 1, 2, 3 \}$ 
is an invariant cluster and an underapproximation of the exact valid set $R_I$. 
if such $p_1$, $p_2$ and $p_3$ are found 
and there is an assignment $(a_0, b_0)$ 
such that $p_i(a_0, b_0) \leq 0$ for $i = 1, 2, 3$, 
then this assignment $(a_0, b_0)$ will be a valid assignment 
which makes $C_1$, $C_2$, $C_3$ satisfiable. 
Consequently, it can be used to extracted a valid invariant 
when substituted back to the invariant template $\textit{Inv}$. 

As shown in Section~\ref{sec:underapproximation}, 
those $p_i$ can be obtained by solving a series 
of sum-of-squares (SOS) relaxations of increasing degrees.
By solving these SOS programs numerically using SDP, 
we obtain a series of valid $p_i$ of increasing degrees. 
$p_i$ with higher degrees give tighter underapproximations 
but would need more computation efforts to obtain. 

In this particular example, the $p_i$ obtained 
by solving the degree $2$ SOS relaxation is: 
\begin{equation*}
\begin{aligned}
    p_1(a, b) = \vphantom{1} &0.9999105909378522 b + 0.9897744745712598 \, , \\
    p_2(a, b) = \vphantom{1} &0.18645935858312943 a^2 + 0.925510976100242 a + \phantom{1} \\
                &1.0000005956843994 b + 3.803019854318091 \, ,  \\
    p_3(a, b) = \vphantom{1} &0.9999979891677253 b - 4.009995082995367 \, .
\end{aligned}
\end{equation*}

\subsection{Using the Invariant Cluster} 
Now $(a_0, b_0) \mathbin{{:}{=}} (-2.000000, -3.959980)$ 
can be extracted from the above invariant cluster $R$ 
using numeric solvers as a valid assignment, 
which gives the following invariant candidate $\phi(x, y) = x^2 -2.000000 y^2 - 3.959980 \leq 0$.
In order to avoid the possible unsoundness caused by the numerical computation, 
the inductive invariance of this invariant candidate is verified by the SMT solver \texttt{Z3}. 

Furthermore, when the specifications and programs are changed slightly, 
the obtained invariant cluster enables us to find new invariants easily. 
For example, if the postcondition is changed to $\{x^2 - 2 y^2 \leq 3 \}$, 
then the above invariant $\phi(x, y)$ is no longer capable of proving the new postcondition. 
However, a new invariant $\phi'(x, y) = x^2 - 2 y^2 - 2.95383593 \leq 0$ 
can be directly extracted from the invariant cluster $R$ 
to re-establish the proof, without doing the invariant synthesis procedure from scratch. 

 \section{Preliminaries}
\label{sec:preliminaries}

\subsection{Basic Notions}
The following basic notations will be used throughout the rest of this paper: 
$\mathbb{R}$, $\mathbb{R}^+$ and $\mathbb{N}$ respectively stand for 
the set of real numbers, the set of positive real numbers 
and the set of non-negative integers. 
We use boldface letters to denote vectors (such as $\bm{x}$, $\bm{y}$) 
and vector-valued functions (such as $\bm{f}(\bm{x})$, $\bm{g}(\bm{x})$). 
If not explicitly stated otherwise, the comparison between vectors is pointwise 
(i,e, for $\bm{x}, \bm{y} \in \mathbb{R}^n$, $\bm{x} \geq \bm{y}$ 
means $x_1 \geq y_1 \land x_2 \geq y_2 \land \dots \land x_n \geq y_n$). 
$U(\bm{x}_0, \delta)$ denotes the $\delta$ neighbor of $\bm{x}_0$, 
i.e. $U(\bm{x}_0, \delta) = \{ \bm{x} \mid \norm{\bm{x} - \bm{x}_0} \leq \delta\}$. 
$\mathbb{R}[\cdot]$ denotes the ring of polynomials 
in variables given by the argument, 
$\mathbb{R}^d[\cdot]$ denotes the set of polynomials 
of degree less than or equal to $d$ 
in variables given by the argument, where $d \in \mathbb{N}$. 
Similarly, $\Sigma[\cdot]$ and $\Sigma^{d}[\cdot]$ denote 
the set of sum-of-squares polynomials 
and sum-of-squares polynomials of degree less than or equal to $d$, respectively. 
For convenience, we do not explicitly distinguish 
a polynomial $p \in \mathbb{R}[\bm{x}]$ and the function it introduces.

A basic semialgebraic set $\mathcal{A}$ is of the form 
$\{\bm{x} \mid p_1(\bm{x}) \triangleright 0 \wedge \ldots \wedge p_n(\bm{x}) \triangleright 0\}$, 
where $p_i(\bm{x}) \in \mathbb{R}[\bm{x}]$, $\triangleright \in \{\geq, >\}$. 
The set 
\begin{equation*}
    \bm{M}(p_1, p_2, \ldots, p_n) := 
        \{ \sigma_0 + \Sigma_{i=1}^{n} \sigma_i p_i \mid \sigma_0, \sigma_i \in \Sigma[\bm{x}] \}
\end{equation*} 
is called the quadratic module generated by $p_1, p_2, \ldots, p_n$. 
The quadratic module $\bm{M}$ is called \emph{Archimedean}, or satisfies \emph{Archimedean condition}, 
if $C - \norm*{\bm{x}}^2 \in \bm{M}$ for some real number $C > 0$. 
Note that if a bound $C$ of $\mathcal{A}$ such that 
$\forall \bm{x} \in \mathcal{A} \ldotp \norm*{\bm{x}}^2 \leq C$ is known, 
we can make the corresponding quadratic module Archimedean 
by simply adding the (redundant) constraint $C - \norm*{\bm{x}}^2 \geq 0$ to the definition of $\mathcal{A}$. 
A semialgebraic set is of the form $\bigcup_{i=1}^{n} \mathcal{S}_i$, 
where $\mathcal{S}_i$ are basic semialgebraic sets.

In order to better illustrate our main ideas, 
in this paper, we focus on the kind of programs presented in Code~\ref{code:model}. 
The \texttt{*} in loop guard means an indeterministic choice. 
Furthermore, the control guards $\bm{c}_i$ are are assumed to be non-intersecting 
(i.e. there is no state $\bm{x}$ satisfies 
both $\bm{c}_i(\bm{x}) \leq 0$ and $\bm{c}_j(\bm{x}) \leq 0$ for any $i \neq j$).
\footnote{
As readers may find out in the rest of this paper, 
the approach presented here can be extended 
to more complex programs (e.g. with non-trivial loop guards, overlapped control flow) 
without essential difficulties. 
}
The safety property needed to be proved is that 
for any state satisfying the precondition $\mathit{Pre}(\bm{x}) \leq 0$, 
if the loop terminates, the exit state must satisfy the postcondition $\mathit{Post}(\bm{x}) \leq 0$.

\begin{listing}[!ht]
\begin{minted}[mathescape, escapeinside=||]{c}
    // Precondition: {$\mathit{Pre}(\bm{x}) \leq 0$} 
    while (*) {
        if (|$\bm{c}_1(\bm{x}) <= 0$|) { |$\bm{x} = \bm{f}_1(\bm{x})$|; continue; }
        if (|$\bm{c}_2(\bm{x}) <= 0$|) { |$\bm{x} = \bm{f}_2(\bm{x})$|; continue; }
        |$\cdots$|
        if (|$\bm{c}_k(\bm{x}) <= 0$|) { |$\bm{x} = \bm{f}_k(\bm{x})$|; continue; }
        break; 
    }
    // Postcondition: {$\mathit{Post}(\bm{x}) \leq 0$}
\end{minted}
\caption{The Program Model}
\label{code:model}
\end{listing}

We additionally assume during the entire run of the program, 
the program state $\bm{x}$ stays in some known compact set $C_{\bm{x}}$. 
For many real-world programs, there is a natural bound that can be used. 
Furthermore, in most practical programming languages (such as \texttt{C}), 
variable values are in fact bounded in a known range.

\subsection{Invariants and Invariant Clusters}
Now we give the formal definitions regarding invariants and invariant clusters. 
\begin{definition}[Invariaint]
\label{def:inv}
    $\mathit{Inv} \subseteq \mathbb{R}^n$ is an invariant of the program in Code~\ref{code:model}
    if it satisfies the following conditions:
    \begin{enumerate}
        \item $(\mathit{Pre}(\bm{x}) \leq 0) \implies (\bm{x} \in \mathit{Inv})$;
        \item $(\bm{c}_i(\bm{x}) \leq 0 \land \bm{x} \in \mathit{Inv}) \implies 
            (\bm{f}_i(\bm{x}) \in  \mathit{Inv})$, \quad $i=1, \dots, k$;
        \item $(\bigwedge_{i=1}^{k} (\bm{c}_i(\bm{x}) \geq  0) \land \bm{x} \in \mathit{Inv}) \implies 
            (\mathit{Post}(\bm{x}) \leq 0)$.
    \end{enumerate}
\end{definition}

\begin{remark}
    In some previous works, the last condition of Definition~\ref{def:inv} 
    is given using $c_i(\bm{x}) > 0$ instead of $c_i(\bm{x}) \geq 0$. 
    However, since numeric methods are used in our approach, 
    it would be unnecessary and unrealistic to distinguish ``strictly greater than'' 
    such as $c_i(\bm{x}) > 0$ from ``greater than'' such as $c_i(\bm{x}) \geq 0$. 
    Therefore, in order to treat these three conditions uniformly, 
    we relaxed it to $c_i(\bm{x}) \geq 0$ as in Definition~\ref{def:inv}. 
    Note that if the third condition holds for $c_i(\bm{x}) \geq 0$, 
    it will also holds for $c_i(\bm{x}) > 0$.
\end{remark}

Clearly, the existence of an invariant implies the safety property to be proved.

Next, we give the definition of polynomial templates used in this paper:  
\begin{definition}[Polynomial Template]
\label{def:template}
    A polynomial template $I(\bm{a}, \bm{x}):C_{\bm{a}} \times \mathbb{R}^n \mapsto \mathbb{R}$ 
    is a polynomial in $\mathbb{R}[\bm{a},\bm{x}]$, 
where $C_{\bm{a}}$ is a known compact subset of $\mathbb{R}^m$. 
    $\bm{a} = (a_1, a_2, \dots, a_m) \in C_{\bm{a}}$ in $I(\bm{a}, \bm{x})$ 
    are referred as parameters of $I(\bm{a}, \bm{x})$.  
\end{definition}

\begin{remark}
    For simplicity, we firstly consider polynomial templates. 
The detailed discussion of basic semialgebraic 
    (and general semialgebraic) templates is given in Section~\ref{sec:extension}. 
\end{remark}

\begin{remark}
    In Definition~\ref{def:template}, 
    the parameter $\bm{a}$ is assumed to be taken from a compact set $C_{\bm{a}}$.
    This is without loss of generality 
    if the template polynomial $I(\bm{a}, \bm{x})$ is homogeneous in $\bm{a}$ 
    (e.g. when the template is taken as $\sum_{\alpha} a_{\alpha} \bm{x}^{\alpha}$ 
    which includes all monomials $\bm{x}^{\alpha}$ under a certain degree $d$). 
    In this case, the parameter $\bm{a}$ can be scaled by any positive constant 
    without changing the invariant candidate it defines. 
    As a result, we may just take $C_{\bm{a}}$ to be $[-1, 1]^m$. 
\end{remark}

An invariant cluster is defined as a subset of invariants 
that can be described using a template and a set of valid parameters: 
\begin{definition}[Invariant Clusters]
    \label{def:clusters}
    An invariant cluster $C$ of a program $P$ w.r.t. a polynomial template $I(\bm{a}, \bm{x})$ 
    is a set of invariants of $P$ given by $\{ I(\bm{a}, \bm{x}) \leq 0 \mid \bm{a} \in R \}$
    for a specified set of parameters $R$. 
\end{definition}

Obviously, the elements of $R$ should produce valid invariant 
after being substituted back to the template. 
The related concepts are formalized below: 
\begin{definition}[Valid Set]
\label{def:valid}
    Given a program $P$ and a template $I(\bm{a},\bm{x})\in \mathbb{R}[\bm{a},\bm{x}]$, 
    a parameter assignment $\bm{a}_0 \in C_{\bm{a}}$ is valid if 
    its instantiation $\{ \bm{x} \mid I(\bm{a}_0,\bm{x}) \leq 0 \}$ 
    is an invariant of the program $P$. 
    The valid set, denoted as $R_I$, is the set of all valid parameter assignments 
    for the polynomial template $I(\bm{a},\bm{x})$. 
\end{definition}

Clearly, the set $R$ describing an invariant cluster 
must be a subset of the valid set $R_I$. 
As simpler and larger $R$ are preferred in most cases, 
the problem of synthesizing invariant clusters renders to 
searching for simpler and tighter \emph{underapproximations} of the valid set $R_I$. 

\subsection{Sum-of-squares Relaxations}
In this subsection, we give a brief introduction to sum-of-squares relaxations 
used in polynomial optimization problems of the following form \eqref{eqn:pop}. 

\begin{mini}|s|
    {\mathclap{\substack{\bm{u}=(u_1,\ldots,u_r)}}}{\bm{c}^T\bm{u}}
    {\label{eqn:pop}}
    {}
    \addConstraint{a_{10}(\bm{x}) + a_{11}(\bm{x})u_1 + \dots + a_{1r}(\bm{x})u_r}{\geq 0}{\quad \forall \bm{x} \in K_1,}
    \addConstraint{\cdots \quad}{\cdots}{,}
    \addConstraint{a_{s0}(\bm{x}) + a_{s1}(\bm{x})u_1 + \dots + a_{sr}(\bm{x})u_r}{\geq 0}{\quad \forall \bm{x} \in K_s,}
\end{mini}
where
\begin{equation*}
    K_i= \{ \bm{x} \in \mathbb{R}^{n} \mid g_{i1}(\bm{x}) \leq 0, \dots, g_{im_i}(\bm{x}) \leq 0 \}, i=1,\ldots,s,
\end{equation*}
and $a_{ij} \in \mathbb{R}[\bm{x}], i=1,\ldots,s;j=1,\ldots,r$ are known polynomials.

The optimization \eqref{eqn:pop} is a polynomial optimization 
with a linear objective function over decision variables $u_i$ 
and some non-negative constraints on certain polynomials. 
Such a constraint demands that 
when $u_i$ are used to linearly combined some known polynomials, 
the resulted polynomials are non-negative on some known basic semialgebraic set $K_i$.

By exploiting the relation between non-negative polynomials and sum-of-squares polynomials, 
some efficient methods have been proposed to solve this type of optimization problems. 
In particular, based on Putinar's Positivstellensatz, 
Lasserre \cite{lasserre2000global} showed that 
we can use the following hierarchy of sum-of-squares relaxations 
(every choice of natural number $d$ corresponds to a SOS programming problem): 

\begin{mini}|s|
    {\mathclap{\substack{\bm{u}=(u_1,\ldots,u_r)}}}{\bm{c}^T\bm{u}}
    {\label{eqn:sosrelaxe}}
    {}
    \addConstraint{a_{10}(\bm{x}) + \sum_{j=1}^{r}a_{1j}(\bm{x})u_j}{= \sigma_{10}(\bm{x}) + \sum_{j=1}^{m_1} \sigma_{1j}(\bm{x})g_{1j}(\bm{x}),}
    \addConstraint{\cdots \quad}{\cdots,}
    \addConstraint{a_{s0}(\bm{x}) + \sum_{j=1}^{r}a_{sj}(\bm{x})u_j}{= \sigma_{s0}(\bm{x}) + \sum_{j=1}^{m_1} \sigma_{sj}(\bm{x})g_{sj}(\bm{x}),}
\end{mini}
where 
\begin{equation*}
    \sigma_{ij} \in \Sigma^{2d}[\bm{x}],i=1,\dots,s; j=1, \dots, m_i\
\end{equation*}
to approximate the optimal solution when a constraint 
of the form $\|\bm{x}\|^2 - M_i \leq 0$ is included in the definition of $K_i$, $i=1, \ldots, s$. 
The sum-of-squares relaxation \eqref{eqn:sosrelaxe} can then be reduced 
to a semidefinite programming problem and be solved efficiently 
(e.g. by interior-point methods) in polynomial time, 
given a desired numeric error bound.

 \section{Underapproximation of the Valid Set}
\label{sec:underapproximation}
In this section, we briefly introduce 
how to obtain underapproximations of the valid set $R_I$. 

\subsection{Representing $R_I$ as Intersections}
The first step of our approach is to 
give a formal and exact description of the valid set $R_I$. 
More specifically, the conditions in Definition~\ref{def:inv} 
are translated one-by-one to some sets $R_I(i)$, 
where each $R_I(i)$ corresponds to the $i$-th condition in Definition~\ref{def:inv}. 

Combining Definition~\ref{def:inv} and Definition~\ref{def:valid}, 
it is easily obtained that 
\begin{equation}
    \label{eqn:defRi}
    \begin{aligned}
        &R_I(0) = \{ \bm{a} \in C_{\bm{a}} \mid \forall \bm{x} \in C_{\bm{x}} \ldotp 
        (\mathit{Pre}(\bm{x}) \leq 0) {\implies} (I(\bm{a}, \bm{x}) \leq 0) \} \, , \\
        &R_I(i) = \{ \bm{a} \in C_{\bm{a}} \mid \forall \bm{x} \in C_{\bm{x}} \ldotp 
((c_i(\bm{x}) \leq 0) \land (I(\bm{a}, \bm{x}) \leq 0))  
                \implies I(\bm{a}, \bm{f}_i(\bm{x})) \leq 0) \} \text{ for } i=1, \dots, k  \, ,  \\
        &R_I(k+1) = \{ \bm{a} \in C_{\bm{a}} \mid \forall \bm{x} \in C_{\bm{x}} \ldotp
(\wedge_{i=1}^{k} (c_i(\bm{x}) \geq 0) \land  
            (I(\bm{a}, \bm{x}) \leq 0)) \implies \mathit{Post}(\bm{x}) \leq 0 \} \, .
    \end{aligned}
\end{equation}
and $R_I  = \bigcap_{i=0}^{k+1} R_I(i)$. 

Note that underapproximations of $R_I$ can be obtained 
by underapproximating each $R_I(i)$ and take the intersection. 
In fact, it is easy to show that 
\begin{proposition}
    Let $S = \bigcap_{i=1}^{n} S(i)$ and $S_u(i) \subseteq S(i)$, 
    then $S_u = \bigcap_{i=1}^{n} S_u(i)$ satisfies $S_u \subseteq S$ 
    and $S \setminus S_u \subseteq \bigcup_{i=1}^n (S(i) \setminus S_u(i))$. 
\end{proposition}
\begin{proof}
    Simply apply the De Morgan's laws. 
\end{proof}

\subsection{Underapproximating $R_I(i)$ using SDP}
The representation of $R_I(i)$ in \eqref{eqn:defRi} involves quantifiers 
and is therefore harder to reason about. 
For example, it is non-trivial even to decide 
whether $\bm{a}_0 \in R_I(i)$ for a given $\bm{a}_0$. 
However, it is easy to find that 
they have a similar structure: 
\begin{equation*}
    R_I(i) = \{ \bm{a} \in C_{\bm{a}} \mid 
    \forall \bm{x} \in \mathbf{K}_i(\bm{a}) \ldotp 
    l_i(\bm{a}, \bm{x}) \leq 0\}
\end{equation*}
where $l_i(\bm{a}, \bm{x})$ is a polynomial and  
$\mathbf{K}_i(\bm{a}) = \{ \bm{x} \in C_{\bm{x}} \mid \bigwedge_{r=1}^{t_i} g_{ir}(\bm{a}, \bm{x}) \leq 0 \}$
for some polynomials $g_{ir}(\bm{a}, \bm{x})$. 
\footnote{
    One may observe that $\mathbf{K}_0(\bm{a})$ and $l_{k+1}(\bm{a}, \bm{x})$ 
    can be defined without using $\bm{a}$ variables and relative operations can be simplified. 
    Nevertheless, we choose to stick to the above structure for unification. 
}

This structure enables us to use techniques presented in \cite{Lasserre15} 
to underapproximate $R_I(i)$. 
We briefly recap its ideas here for self-containment.  

First, notice that 
\begin{equation*} 
    R_I(i) = \{ \bm{a} \in C_{\bm{a}} \mid \bar{J}_i(\bm{a}) \leq 0 \}
\end{equation*}
where $\bar{J}_i(\bm{a}) = \sup_{\bm{x} \in \mathbf{K}_i(\bm{a})} l_i(\bm{a}, \bm{x})$. 
We additionally define $J_i(\bm{a}) = \max (\bar{J}_i(\bm{a}), -M)$ for some real number $M > 0$ 
(we choose $M = 10.0$ in experiments). 
It is obvious that $J_i(\bm{a}) \leq 0$ and $\bar{J}_i(\bm{a}) \leq 0$ define the same set.

\begin{remark}
    Directly using the function $\bar{J}_i(\bm{a})$ will involve difficulties 
    when $\mathbf{K}_i(\bm{a})$ is empty (in which case, $\bar{J}_i(\bm{a})$ becomes $-\infty$). 
    Therefore, \cite{Lasserre15} assume $\mathbf{K}_i(\bm{a})$ 
    to be non-empty in all the following results. 
    However, we observe that by using $J_i(\bm{a})$ instead of $\bar{J}_i(\bm{a})$, 
    the condition $\mathbf{K}_i(\bm{a}) \neq \emptyset$ can be dropped. 
\end{remark}

Underapproximating the set $R_I(i)$ can therefore be done by 
approximating $J_i(\bm{a})$ from above. 
The function $J_i(\bm{a})$ can be shown 
to be upper-semicontinuous in $C_{\bm{a}}$ 
(i.e. for all $\bm{a}_0 \in C_{\bm{a}}$,  
$\displaystyle \limsup_{\bm{a} \to \bm{a}_0} J_i(\bm{a}) \leq J_i(\bm{a}_0)$
holds)
similarly to \cite[Lemma~1]{Lasserre15}, 
and therefore admits effective approximations using polynomials. 
In particular, a series of polynomial approximations can be obtained 
by solving the following hierarchy of SOS programs:

\begin{mini}|s|
    {\mathclap{
            \substack{
                \mathrm{p}_{\alpha}, \sigma_{i}, \sigma^{K}_{ir}, \sigma^{\bm{a}}_{ir}, \\ \sigma^{\bm{x}}_{ir}, \sigma'_{i}, \sigma'^{K}_{ir}, \sigma'^{\bm{a}}_{ir}, \sigma'^{\bm{x}}_{ir}, }
        }
    }{\sum_{\alpha} \gamma_{\alpha}\mathrm{p}_{\alpha}^i}
    {\label{eqn:sosrelax}}
    {\bm{M}_d(i):}
    \addConstraint{p_i(\bm{a}) {-} l_i(\bm{a}, \bm{x}) {=} \sigma_{i}(\bm{a}, \bm{x}) {-} \sum_{r=1}^{m_i} \sigma^K_{ir}(\bm{a}, \bm{x})g_{ir}(\bm{a}, \bm{x}) } {} 
    \addConstraint{\phantom{1} - \sum_{r=1}^{s_{\bm{a}}}\sigma^{\bm{a}}_{ir}(\bm{a}, \bm{x}) h^{\bm{a}}_{r}(\bm{a}) - \sum_{r=1}^{s_{\bm{x}}} \sigma^{\bm{x}}_{ir}(\bm{a}, \bm{x}) h^{\bm{x}}_r(\bm{x}),} \addConstraint{p_i(\bm{a}) - M {=} \sigma'_{i}(\bm{a}, \bm{x}) {-} \sum_{r=1}^{m_i} \sigma'^K_{ir}(\bm{a}, \bm{x})g_{ir}(\bm{a}, \bm{x}) } {} 
    \addConstraint{\phantom{1} - \sum_{r=1}^{s_{\bm{a}}}\sigma'^{\bm{a}}_{ir}(\bm{a}, \bm{x}) h^{\bm{a}}_{r}(\bm{a}) - \sum_{r=1}^{s_{\bm{x}}} \sigma'^{\bm{x}}_{ir}(\bm{a}, \bm{x}) h^{\bm{x}}_r(\bm{x}),} \end{mini}
where 
$\mathrm{p}_{\alpha}^i \in \mathbb{R}$, 
$\sigma_{i}, \sigma^K_{ir}, \sigma^{\bm{a}}_{ir}, \sigma^{\bm{x}}_{ir}, \sigma'_{i}, \sigma'^K_{ir}, \sigma'^{\bm{a}}_{ir}, \sigma'^{\bm{x}}_{ir} \in \Sigma[\bm{a}, \bm{x}]^{2d}$. 
$p_i(\bm{a}) = \sum_{\alpha} \mathrm{p}_{\alpha}^i \bm{a}^{\alpha}$ 
is a polynomial that includes all monomials up to degree $2d$. 
$\gamma_{\alpha}$ are rescaled moments defined as
\begin{equation*}
    \gamma_{\alpha} =\frac{1}{\mu(C_{\bm{a}})}\int_{C_{\bm{a}}} \bm{a}^{\alpha} \mathrm{d} \mu(\bm{a}) \, , 
\end{equation*}
where $\mu(\cdot)$ denotes the Lebesgue measure. 
$h_r^{\bm{a}}(\bm{a})$, $h_r^{\bm{x}}(\bm{x})$ are chosen such that
\begin{equation*}
\begin{aligned}
    &C_{\bm{a}} = \{ \bm{a} \mid \bigwedge_{r=1}^{s_{\bm{a}}} h^{\bm{a}}_r(\bm{a}) \leq 0 \} \, ,  \\[-1mm]
    &C_{\bm{x}}= \{ \bm{x} \mid \bigwedge_{r=1}^{s_{\bm{x}}} h^{\bm{x}}_r(\bm{x}) \leq 0 \} \, ,  \\[-1mm]
\end{aligned}
\end{equation*}
and the corresponding quadratic modules 
of $\mathbf{K}_i(\bm{a})$, $C_{\bm{a}}$ and $C_{\bm{x}}$ are assumed to be Archimedean. 
\footnote{ 
    This can be done by simply adding a redundant ball constraint in their definition. 
}

Given a set of feasible assignments of $\mathrm{p}_{\alpha}^i \in \mathbb{R}$ 
of the degree $d$ relaxation $\bm{M}_d(i)$, 
the polynomial approximation of $J_i(\bm{a})$ can be obtained as 
$p^{(d)}_i(\bm{a}) = \sum_{\alpha} \mathrm{p}^i_{\alpha} \bm{a}^{\alpha}$ 
and the corresponding underapproximation of $R_I(i)$ 
is $R^{(d)}(i) = \{ \bm{a} \in C_{\bm{a}} \mid p^{(d)}_i(\bm{a}) \leq 0 \}$. 

Much like what has been done in \cite{Lasserre15}, 
the underapproximations $R^{(d)}(i)$ can be proved to have many desired properties, including: 

\begin{theorem}[Soundness]
\label{thm:soundness}
    Given a feasible solution of \eqref{eqn:sosrelax} and $R^{(d)}(i)$ obtained as above,  
    $R^{(d)} = \bigcap_{i=0}^{k+1} R^{(d)}(i)$ is an underapproximation 
    of the valid set $R_I$, i.e. $R^{(d)} \subseteq R_I$. 
\end{theorem}

\begin{proof}
    Very similar to \cite[Theorem~3]{Lasserre15}.
\end{proof}

\begin{theorem}[Convergence]
    \label{thm:convergence}
    If for every degree $d$ the program $\mathbf{M}_d(i)$ is solvable, 
    assume $R^{(d)}$ is built as above using the optimal solution, 
    then 
    \begin{equation}
    \label{eqn:convergenceRd}
        \lim_{d \to \infty} \mu(R_I(i) \setminus R^{(d)}(i)) = 0, \forall i\in {0,1,\ldots, k+1}.
    \end{equation}
    provided that the set 
    $R_I^r(i) = \{ \bm{a} \in C_{\bm{a}} \mid J_i(\bm{a}) = 0 \}$
    has Lebesgue measure zero. 
\end{theorem}

\begin{proof}
    Very similar to \cite[Theorem~5]{Lasserre15}.
\end{proof}

\begin{remark}
    In \cite[Theorem~5]{Lasserre15}, it has been shown that 
    if the quadratic module corresponding to the constraints of \eqref{eqn:sosrelax} 
    satisfies Archimedean condition and the feasible region contains an interior point, 
    then the SOS problem \eqref{eqn:sosrelax} is solvable. 
Moreover, as shown in \cite[Theorem~1]{Josz2016}, 
    we can also avoid to check the existence of an interior point 
    by adding the (redundant) constraints 
    $\|\bm{x}\|^2 \leq M_{\bm{x}}, \|\bm{a}\|^2 \leq M_{\bm{a}}$ to \eqref{eqn:sosrelax}, 
    which will also guarantee that the SOS problem \eqref{eqn:sosrelax} is solvable. 
\end{remark}

\begin{remark}
    The assumption that $R_I^r(i)$ has Lebesgue measure zero 
    basically states that 
    the set of zero points of $J_i$ should be negligible. 
    Note that the zero points of $J_i$ are exactly zero points of $\bar{J}_i$. 
    We show in the following Lemma~\ref{thm:semiAlge} that 
    $\bar{J}_i$ is a semialgebraic function. 

    \begin{lemma}[Semialgebraic functions]
    \label{thm:semiAlge}
        Let $A \subseteq B\subseteq \mathbb{R}^{m+n}$ be two semialgebraic sets. 
        Let $\psi: B \to \mathbb{R}$ be a polynomial function. 
        Then $\theta(\bm{a})=\sup_{(\bm{a}, \bm{x})\in A} \psi(\bm{a},\bm{x})$ is a semialgebraic function. 
        In particular, $\bar{J}_i$ are semialgebraic functions.
    \end{lemma}

    \begin{proof}
        By definition of semialgebraic functions, 
        we only need to show the graph of $\theta(\bm{a})$ is a semialgebraic set 
        (see e.g. \cite[Definition~2.2.5]{bochnak1998real}). 
        The graph of the function $\theta(\bm{a})$ is
        \begin{equation*}
        \begin{aligned}
            \{
            (\bm{a},y)\in \mathbb{R}^{m+1} \mid 
            &(\forall \bm{x} \ldotp (\bm{a}, \bm{x}) \in A \implies \psi(\bm{a},\bm{x})\leq y) \land \\
            &(\forall \epsilon\in \mathbb{R}^+, \exists \bm{x}. (\bm{a}, \bm{x}) \in A \land
	    \psi(\bm{a},\bm{x})+\epsilon > y)
	    \},
        \end{aligned}
        \end{equation*}
        which is a semialgebraic set in $\mathbb{R}^{m+1}$ by Tarski-Seidenberg principle
        (see e.g. \cite[Definition~2.2.3, Proposition~2.2.4]{bochnak1998real}).

        Applying this result to $\phi''_i$ with 
        $A = \{ (\bm{a}, \bm{x}) \in C_{\bm{a}} \times C_{\bm{x}} \mid \bm{x} \in K_i(\bm{a}) \}$, 
        $B = \mathbb{R}^{m+n}$ and $\psi = l_i$, 
        it follows that $\phi''_i$ are semialgebraic functions.
    \end{proof}
  
    Now, as $\bar{J}_i$ is semialgebraic, 
    by \cite[Lemma~2.5.2]{bochnak1998real}, 
    there exists a nonzero polynomial $ h \in \mathbb{R}[\bm{a}, y]$ such that 
    $h(\bm{a}, \bar{J}_i(\bm{a})) = 0 $ for every $\bm{a}\in C_{\bm{a}}$. 
    The set of zero points of $\bar{J}_i(\bm{a})$ 
    is therefore contained in the zero level set of a polynomial $h(\bm{a}, 0)$. 

    As $h(\bm{a}, 0)$ is a polynomial, 
    if the assumption of $R_I^r(i)$ is zero-measured were to be violated, 
    $h(\bm{a}, 0)$ must be constant zero. 
    In other words, $h(\bm{a}, y)$ contains $y$ as a factor, 
    which is relatively rare in practice. 
\end{remark}

Theorem~\ref{thm:convergence} indicates that 
if the set $R_I^r(i)$ is negligible, 
then \emph{almost} all valid assignments can be included 
when the relaxation degree $d$ is high enough. 
Therefore, we obtain the following completeness result: 
\begin{theorem}[Weak Completeness]
\label{thm:weakComplete}
    if the set $R_I^r(i)$ has Lebesgue measure zero, 
    the valid set $R_I$ contains an interior point 
    and the SOS programs $\bm{M}_d(i)$ are solvable, 
    then the above procedure can always find 
    a non-empty underapproximation $R^{(d)}$ of the valid set $R_I$ 
    (and subsequently, find an invariant). 
\end{theorem}

\begin{proof}

    Let ${R^{(d)}=\bigcap_{i=0}^{k+1}R^{(d)}(i)}$. 
    Since $R_I = \bigcap_{i=0}^{k+1} R_I(i)$ has positive Lebesgue measure, 
    \eqref{eqn:convergenceRd} indicates that 
    \begin{equation*}
        \lim_{d \to \infty} \mu(R_I \setminus R^{(d)}) = 0,
    \end{equation*}
    which further implies that $R^{(d)}$ 
    has positive Lebesgue measure when $d$ is large enough. 
    This $R^{(d)}$ is therefore a non-empty and 
    is an underapproximation of the valid set $R_I$ by Theorem~\ref{thm:soundness}.
\end{proof}

 \section{Additional Techniques} 
\label{sec:techniques}
The methods presented in Section~\ref{sec:underapproximation} 
has been shown to have some desirable theoretical properties 
(cf. Theorem~\ref{thm:soundness}, Theorem~\ref{thm:convergence} and Theorem~\ref{thm:weakComplete}). 
However, it still faces some challenges when directly applied in practice. 
One key problem is that the solving of SOS program $\mathbf{M}_d(i)$ 
becomes difficult when the relaxation degree $d$ is relatively high. 
In this section, we present some additional techniques 
that can be used to ease the problem and improve the efficiency. 
In addition, we also discuss the potential unsoundness induced by numerical errors 
and possible ways to ease its influence.

\subsection{Adaptive Partitioning} 
In Section~\ref{sec:underapproximation}, 
one of the first steps when underapproximating the valid set 
is to fix a compact set $C_{\bm{a}}$ 
containing possible choices of parameters. 
The underapproximation procedure then amounts to 
approximating some (upper-semicontinuous) functions $J_i(\bm{a})$ 
from above using degree $2d$ polynomials $p^{(d)}_i$ in the set $C_{\bm{a}}$. 
The obvious way to improve the precision of the approximations 
is to increase the relaxation degree $d$. 
However, one may also improve the precision by using a smaller $C_{\bm{a}}$. 
This subsection explores this idea 
and proposes an adaptive partitioning scheme that partitions $C_{\bm{a}}$ 
to improve the approximation precision with polynomials. 

Given the parameter set $C_{\bm{a}}$ 
and a partition $(C'_{\bm{a}}, C''_{\bm{a}}) = \mathit{Bisec}(C_{\bm{a}})$, 
it can be observed that the valid set $R_I$ w.r.t. $C_{\bm{a}}$ 
is the union of the valid set $R'_I$ and $R''_I$
w.r.t. $C'_{\bm{a}}$ and $C''_{\bm{a}}$ respectively, 
and underapproximations of $R_I$ can be obtained by 
taking the union of underapproximations of $R'_I$ and $R''_I$. 
One may therefore design a simple recursive procedure accordingly. 

The problem remains here is to decide when to stop the partitioning. 
One common standard is to stop partitioning when the current $C_{\bm{a}}$ is small enough. 
On top of that, we give an additional, adaptive standard in the following 
by looking closely into the underapproximation procedure. 

First, notice that the objective value of $\mathbf{M}_d(i)$ (denoted as $v$ here)
is actually the rescaled integral of the polynomial $p^{(d)}_i(\bm{a})$ 
on the current $C_{\bm{a}}$ (see \cite{Lasserre15} for details). 
As $p^{(d)}_i(\bm{a})$ approximates $J_i(\bm{a})$ from the above, 
it can be seen that $v$ measures
how good the polynomial $p^{(d)}_i(\bm{a})$ approximates $J_i(\bm{a})$ to some extents. 
Furthermore, it is easy to show that 
if $C_{\bm{a}}$ is partitioned as $C'_{\bm{a}}$ and $C''_{\bm{a}}$ 
and the respective objective values $v'$ and $v''$ are computed by solving \ref{eqn:sosrelax}, 
there will always be $\Delta v = v - \frac{1}{2} (v' + v'') \geq 0$. 
This difference $\Delta v$ can subsequently be used 
to measure the improvements induced by this partitioning step. 
A large $\Delta v$ indicates that probably more partitioning is needed, 
whereas a small $\Delta v$ indicate that the room for improvements 
by further partitioning is probably limited, and the partitioning process could be stop. 

In summary, the adaptive partitioning scheme can be described as follows:
\begin{steps}
    \item If the diameter of current $C_{\bm{a}}$ 
          is less than a given threshold $\epsilon_d > 0$, stop. 
          Otherwise, for current $C_{\bm{a}}$, 
          build underapproximations using 
          methods in Section~\ref{sec:underapproximation} 
          and record the optimal value $v$. 
    \item Partition the current $C_{\bm{a}}$ 
          as the union of $C'_{\bm{a}}$ and $C''_{\bm{a}}$, 
          build underapproximations also for $C'_{\bm{a}}$ and $C''_{\bm{a}}$, 
          and record the respective optimal value $v'$ and $v''$. 
          Compute $\Delta v$ using $\Delta v = v - \frac{1}{2} (v' + v'')$. 
    \item If $\Delta v$ is less than a given threshold $\epsilon_v > 0$, stop. 
          Otherwise, recursively apply this procedure on $C'_{\bm{a}}$ and $C''_{\bm{a}}$ 
          and take the union of the respective results as the final underapproximation result. 
\end{steps}

One may also notice that the partitioning 
of the parameter set $C_{\bm{a}}$ here echoes with 
the verification approaches based on interval analysis \cite{djaballah2017construction}. 
Interval analysis can be used in parallel 
with the adaptive partitioning scheme described above, 
which will help in further pruning of partitions from early on. 

\subsection{Alternative Polynomial Basis}
In order to specify the unknown polynomial 
$p_i(\bm{a}) \in \mathbb{R}^{2d}[\bm{a}]$ in \eqref{eqn:sosrelax}, 
real-valued decision variables $\mathrm{p}_{\alpha}^i$ are assigned 
to denote the coefficients for every monomials under degree $2d$, 
and we have $p_i(\bm{a}) = \sum_{\alpha} \mathrm{p}_{\alpha}^i \bm{a}^{\alpha}$. 
However, this is not the only way to specify an unknown polynomial using parameters. 
Such representation essentially connects to the concept of polynomial basis. 

In short, a polynomial basis of a polynomial vector space 
is a set of linearly independent polynomials that spans the space. 
Let $\mathcal{P}$ be the set of polynomials under degree $d$, 
the most common polynomial basis is the monomial basis 
$\{1, \bm{a}, \dots, \bm{a}^{\alpha}, \dots \}$ 
which include all monomials under the degree $d$. 

For the above polynomial vector space $\mathcal{P}$, 
besides the monomial basis, 
there are other non-trivial basis such as Bernstein basis and Chebyshev basis. 
While using a different basis may not change the number of decision variables in \eqref{eqn:sosrelax} 
(since they are just an alternative representation of the same program), 
it may be more stable numerically in practical solving 
(see \cite[Section 3.1.5]{blekherman2012semidefinite}). 

For example, to use Chebyshev basis to reform \eqref{eqn:sosrelax}, 
one can change the definition of $p_i(\bm{a})$
to $p_i(\bm{a}) = \sum_{\alpha} \mathrm{p}_{\alpha}^i t_d^{\alpha}(\bm{a})$ 
where $\{ t_d^{\alpha} \}$ are Chebyshev polynomials of degree $2d$.

\subsection{Numerical Errors in SDP solving} 
In Section~\ref{sec:underapproximation}, 
the underapproximations of the valid set are obtained 
by solving the relaxed SOS programs $\bm{M}_d$, 
which will be ultimately solved as some SDP programs. 
As typical SDP solvers are based on numerical computation, 
there will be inevitably potential numerical errors in the process, 
which could cause potential unsoundness of the result. 
In the following, we present some ways that could be used to ease the effects. 

\begin{itemize}
    \item 
        Posterior verifications: 
        One easy way is to always using an exact symbolic method to check 
        the soundness of the solutions returned by numerical solvers \cite{DXZ13}. 
        Compared with directly solving the constraints, 
        checking the soundness of a certain solution 
        is much easier for symbolic solvers, 
        such as Redlog \cite{dolzmann1997} or Z3 \cite{Z3}. 
        This approach is relatively easy to employ, 
        and can be used \emph{after} the numerical solutions are given. 
        However, for some larger problems, 
        even checking the soundness of a solution symbolically can be difficult. 
    \item 
        More precise SDP solving: 
        Alternatively, one may consider to increase the precision in the process of SDP solving, 
        or even consider an exact SDP solving approach as given in \cite{Henrion18}. 
        Nevertheless, although one may increase the precision of SDP solving 
        by using multiple-precision or arbitrary-precision solvers \cite{nakata2010numerical, JoldesMP17}, 
        the possibility of unsoundness caused by numerical errors can not be eliminated completely. 
        On the other hand, an exact SDP solving 
        would probably resort to symbolic methods (such as \cite{Henrion18}), 
        and can only solve SDP instances of small sizes. 
    \item 
        Validated SDP solving: 
        Finally, one could resort to validated SDP solving proposed by \cite{RVS16}. 
        The basic idea therein is to firstly compute an error bound $\epsilon$ 
        of the numerical errors of the results given by the solvers. 
        After that, the constraints of the original problems are replaced 
        by their $\epsilon$-strengthening versions (e.g. $A \succeq 0$ to $A + \epsilon I \succeq 0$). 
        The solutions given by solving the strengthened version 
        can then be safely used as sound solutions. 
        The ideas are extended and detailed in \cite{Gan2020}, where the authors 
        guarantee the soundness of SDP solving when synthesizing non-linear Craig interpolants. 
        This approach can guarantee the soundness with minimal performance loss, 
        but the strengthening of constraints (i.e. shrinking of the feasible set) 
        means that the completeness results are lost. 
\end{itemize}

In this paper, as the examples used in experiments are relatively small, 
we apply the symbolic posterior verification methods 
to check the soundness of the results given by numerical solvers. 

 \section{Extensions}
\label{sec:extension}

\subsection{Semialgebraic Template}
In this subsection, we discuss the extension 
of our approach to invariant synthesis with semialgebraic templates. 
First, we observe that the techniques introduced in this paper can be applied 
to the cases when templates are basic semialgebraic 
(instead of only polynomial) without substantial changes. 
After that, we briefly discuss the possible application 
of our algorithm to the cases when templates are semialgebraic.

The basic semialgebraic template is formally defined as follows: 
\begin{definition}
    A basic semialgebraic template $\mathit{Inv}_b(\bm{a}, \bm{x})$ is a finite collection 
    of polynomials $I_r(\bm{a}, \bm{x}) :C_{\bm{a}} \times \mathbb{R}^n \mapsto \mathbb{R}$ 
    in $\mathbb{R}[\bm{a}, \bm{x}]$, 
    where $C_{\bm{a}} \subseteq \mathbb{R}^m$ is a known compact set. 
Here, $\bm{a} = (a_1, a_2, \dots, a_m) \in C_{\bm{a}}$ in $I(\bm{a}, \bm{x})$ 
    are referred as parameters of $I(\bm{a}, \bm{x})$.  
    Given a parameter assignment $\bm{a}_0 \in \mathbb{R}^{m}$, 
    the instantiation of $\mathit{Inv}_b$ w.r.t. $\bm{a}_0$ is 
    the set ${\{\bm{x} \mid \bigwedge_{r} I_r(\bm{a}_0, \bm{x}) \leq 0} \}$.
\end{definition}

A brief review of techniques presented in previous sections 
indicates that our algorithm can be extended 
to basic semialgebraic case with only minimal modifications. 
In particular, in the basic semialgebraic case, instead of a single polynomial, 
$l_i(\bm{a}, \bm{x})$ should be changed to a maximal of polynomials 
(for example, $l_0$ should be $\max_{r} I_r(\bm{a}, \bm{x})$). 
The derived SOS programs will be much like \eqref{eqn:sosrelax} 
but will contain multiple constraints. 
After that, all other results can be derived similarly.

As for general semialgebraic templates, 
we show that this case can be treated by \emph{lifting} to higher dimensions. 
This shows a theoretical possibility 
to use our algorithm to synthesize semialgebraic invariants. 

We give formal definition of general semialgebraic templates as follows:
\begin{definition}
    A (general) semialgebraic template $\mathit{Inv}_g(\bm{a}, \bm{x})$ is a finite collection 
    of polynomials $I_{rt}(\bm{a}, \bm{x}) :C_{\bm{a}} \times \mathbb{R}^n \mapsto \mathbb{R}$ 
    in $\mathbb{R}[\bm{a}, \bm{x}]$, 
where $C_{\bm{a}} \subseteq \mathbb{R}^m$ is a known compact set. 
    $\bm{a} = (a_1, a_2, \dots, a_m) \in C_{\bm{a}}$ in $I(\bm{a}, \bm{x})$ 
    are referred as parameters of $I(\bm{a}, \bm{x})$.  
    Given a parameter assignment $\bm{a}_0 \in \mathbb{R}^m$, 
    the instantiation of $\mathit{Inv}_g$ w.r.t. $\bm{a}_0$ is 
    the set ${\{ \bm{x} \mid \bigvee_{r} \bigwedge_{t} I_{rt}(\bm{a}_0, \bm{x}) \leq 0 \}}$.
\end{definition}

The key observation here is that every semialgebraic set 
is the projection of a closed basic semialgebraic set \cite{bochnak1998real}. 
The following lemma details the lifting we needed 
and can be easily proven as e.g. a corollary of \cite[Lemma 14.3]{lasserre2012beyond}. 

\begin{lemma}
\label{lem:lifting}
    Let $C \subset \mathbb{R}^d$ be a compact basic semialgebraic set defined as 
    $C = \{ \bm{x} \in \mathbb{R}^d \mid g_v(\bm{x}) \leq 0, v = 1, \dots, m \}$
    and $\forall \bm{x} \in C \ldotp M - \norm{\bm{x}} \geq 0$ for some known $M > 0$.
    For any semialgebraic set 
    $S = \{ \bm{x} \in C \mid \bigvee_{i=1}^{n} \bigwedge_{j=1}^{m} f_{ij}(\bm{x}) \leq 0 \}$ 
    where $f_{ij} \in \mathbb{R}[\bm{x}]$, there exists a basic semialgebraic lifting. 
    In particular, there exists $p, s \in \mathbb{N}$ 
    and polynomials $h_1, h_2, \dots, h_s \in \mathbb{R}[\bm{x}, y_1, y_2, \dots, y_p]$ such that:
    \begin{equation*}
        S = \{ \bm{x} \in C \mid \exists \bm{y} \ldotp (\bigwedge_{k=1}^{s} h_k(\bm{x}, \bm{y}) \leq 0) \land (y_p \leq 0) \}.
    \end{equation*}
\end{lemma}

Based on Lemma~\ref{lem:lifting}, 
we show that the general semialgebraic cases can be treated 
by lifting to basic semialgebraic cases and a two-step approximation. 
First observe that the set
\begin{equation*}
    \{ (\bm{a}, \bm{x}) \in C_{\bm{a}} \times C_{\bm{x}} \mid 
    \bigvee_{r} \bigwedge_{t} I_{rt}(\bm{a}, \bm{x}) \leq 0 \}
\end{equation*}
has a basic semialgebraic lifting. 
According to Lemma~\ref{lem:lifting}, this set can be written as: 
\begin{equation*}
    \{ (\bm{a}, \bm{x}) \in C_{\bm{a}} \times C_{\bm{x}} \mid 
    \exists \bm{y} \ldotp (\bigwedge_{k=1}^{s} h_k(\bm{a}, \bm{x}, \bm{y}) \leq 0) \land (y_p \leq 0) \}
\end{equation*}
for some polynomials $h_1, h_2, \dots, h_s \in \mathbb{R}[\bm{x}, y_1, y_2, \dots, y_p]$.

As both $\bm{a}$ and $\bm{x}$ have known bounds 
(respectively, $C_{\bm{a}}$ and $C_{\bm{x}}$), 
it can be proved that $\bm{y}$ lies in some known compact set $C_{\bm{y}}$. 
Therefore, the 
$\exists \bm{y} \ldotp (\bigwedge_{k=1}^{s} h_k(\bm{a}, \bm{x}, \bm{y}) \leq 0) \land (y_p \leq 0)$ 
part 
can be approximated by conjunctions 
of some polynomial inequalities $p_k(\bm{a}, \bm{x}) \leq 0$ 
using techniques presented in \cite{Lasserre15}. 
The conjunctions $p_k(\bm{a}, \bm{x}) \leq 0$ can now be treated 
as a basic semialgebraic template on which our algorithm can be applied. 

In summary, it has been shown that our algorithm can be used 
to synthesize general semialgebraic invariants once a \emph{lifting} is provided. 
Regarding how to compute such a lifting, readers may refer to \cite{lasserre2012beyond}. 
Unfortunately, completeness results 
such as Theorem~\ref{thm:weakComplete} are difficult to obtain 
and would surely need stronger assumptions.

In practice, our algorithm is less efficient 
for general semialgebraic cases compared to polynomial and basic semialgebraic cases. 
The main reason lies in the lifting process: 
applying lifting dramatically increases 
either the degree of defining polynomials 
or the number of parameters, sometimes even both. 
The sets of valid parameter assignments 
of lifted templates also tend to have more complex boundaries, 
which means higher relaxation degree $d$ is needed. 

\subsection{Non-polynomial Functions}
When the program of interest contains non-polynomial continuous functions 
(such as exponential, logarithmic or trigonometric functions) in conditionals or assignments, 
it is no longer plausible to directly use the techniques presented in previous sections 
to synthesize invariant clusters, even if the template itself is polynomial. 
The main difficulty lies in the solving 
of the (now non-polynomial) optimization problem to approximate $J_i(\bm{a})$. 
In polynomial cases, the optimization problem 
is solved by relaxed to a hierarchy of SOS programs. 
The resulted SOS relaxations are in turn solved by SDP solvers. 
However, when non-polynomial functions are involved, 
both the relaxation and the solving would be problematic: 
Positivstellensatz theorems for non-polynomial functions \cite{lasserre2012beyond} 
require stronger conditions and are harder to use; 
even in cases when a hierarchy of sum-of-squares relaxations can be built, 
such a problem can no longer be easily cast 
as a SDP problem since non-polynomial functions are involved. 

However, that does not mean that 
there is nothing can be done when non-polynomial functions are present. 
A common way of treating them is by \emph{symbolic abstraction} \cite{reps2012symbolic,LZZZ15}. 
The basic idea is to use a formula in the abstract domain 
(in our case, the conjunctions of polynomial inequalities) 
to best overapproximate the ``real meaning'' of the original (non-polynomial) formula. 

In this subsection, we briefly explain how to combine our approach with symbolic abstraction 
to synthesize invariant clusters where non-polynomial continuous functions are present. 
In the following, 
the precondition $\mathit{Pre}(\bm{x})$, the postcondition $\mathit{Post}(\bm{x})$, 
the conditionals $\bm{c}_i(\bm{x})$ and the assignments $\bm{f}_i(\bm{x})$ in Code~\ref{code:model} 
are assumed to contain terms built by non-polynomial continuous functions. 

Let $t_i$ be the collection of all variables 
and non-polynomial terms resulted from the first application 
of a non-polynomial function to a specific term, 
e.g., $\sin(x+y)$, but not $\sin(x)+\sin(y)$. 
The abstract mapping $\mathit{abt}$ can be defined from bottom up as follows: 
\begin{equation*}
    \begin{aligned}
        \mathit{abt}(c) &= c, c \in R \, , \\
        \mathit{abt}(a_j) &= a_j \, ,  \\
        \mathit{abt}(t_i) &= z_i  \, ,  
    \end{aligned} 
    \hspace*{0.5cm} 
    \begin{aligned}
        \mathit{abt}(u \leq 0) &= \mathit{abt}(u) \leq 0 \, , \\
        \mathit{abt}(u_1 + u_2) &= \mathit{abt}(u_1) + \mathit{abt}(u_2) \, ,  \\
        \mathit{abt}(u_1 * u_2) &= \mathit{abt}(u_1) * \mathit{abt}(u_2) \, ,  \\
        \mathit{abt}(w_1 \land w_2) &= \mathit{abt}(w_1) \land \mathit{abt}(w_2) \, .
    \end{aligned}
\end{equation*} 
where $c$ is a constant, $z_i$ are abstraction variables, 
$u$ denotes a term and $w$ denotes a formula. 

Note that the result of $\mathit{abt}$ is conjunctions 
of \emph{polynomial inequalities} as all non-polynomial terms are abstracted as new variables. 
As variables $\bm{x}$ are taken from a known compact set $C_{\bm{x}}$, 
We can find $M_{\bm{x}}$ such that $\forall \bm{x} \in C_{\bm{x}} \ldotp \norm*{x}^2 \leq M_{\bm{x}}$. 
Since all the non-polynomial functions used in the program are assumed to be continuous, 
it can be concluded that the abstracted variables $\bm{z}$ also fall in a compact set $C_{\bm{z}}$ 
satisfying $\forall \bm{z} \in C_{\bm{z}} \ldotp \norm*{z}^2 \leq M_{\bm{z}}$ for some $M_{\bm{z}} > 0$.

We also define \emph{strengthening} of abstraction of formula $w$ as follows:
\begin{definition}
A strengthening $\mathit{abt}_{st}(w)$ of an abstraction $\mathit{abt}(w)$ is 
a conjunction of polynomial inequalities 
satisfying $w \implies \mathit{abt}_{st}(w)$ as well as
$\mathit{abt}_{st}(w) \implies \mathit{abt}(w)$.
\end{definition}

The best (strongest) strengthening of an abstraction $\mathit{abt}(w)$ 
is difficult to find, and sometimes does not even exist. 
One usually needs to resort to heuristics 
regarding the specific non-polynomial functions being abstracted 
to obtain a good strengthening procedure. 
We do not expand further on this. 
Interested readers may refer to e.g., \cite{reps2012symbolic, kincaid2017nonlinear} for some examples. 
We assume in the following that a strengthening procedure is available. 

As in Section~\ref{sec:underapproximation}, we write $R_I$ as the intersection of $R_I(i)$. 
Using abstraction mapping $\mathit{abt}$, 
we can define the following \emph{abstracted valid set} $R_a(0)$ and $R_a(k+1)$ as follows:
\begin{equation*}
\begin{aligned}
    R_a(0) &= \{ \bm{a} \in C_{\bm{a}} \mid \forall \bm{z} \in C_{\bm{z}} \ldotp 
\mathit{abt}_{st}(\mathit{Pre}(\bm{x}) \leq 0) 
        \implies \mathit{abt}(I(\bm{a}, \bm{x}) \leq 0) \} \, ,  \\
R_a(k+1) &= \{ \bm{a} \in C_{\bm{a}} \mid \forall \bm{z} \in C_{\bm{z}} \ldotp 
        \mathit{abt}_{st}(\wedge_{i=1}^k(\bm{c}_i(\bm{x}) \geq 0) \land (I(\bm{a}, \bm{x}) \leq 0)) 
        \implies \mathit{abt}(\mathit{Post}(\bm{x}) \leq 0) \} \, . 
\end{aligned}
\end{equation*} 
Note the use of $\mathit{abt}$ and $\mathit{abt}_{st}$ in different positions in the formulas. 

It can be proved that $R_a(0) \subseteq R_I(0)$ and $R_a(k+1) \subseteq R_I(k+1)$. 
In order to deal with the remaining $R_I(i)$, 
two sets of fresh variables $\{z_i\}$ and $\{z'_i\}$ are needed. 
They represent respectively the abstraction of program states 
before and after the loop executes. 
The abstracted valid set $R_a(i)$ can subsequently be defined as:
\begin{equation*}
\begin{aligned}
    R_a(i) = \{ \bm{a} \in C_{\bm{a}} \mid 
        &\forall \bm{z} \in C_{\bm{z}} \forall \bm{z'} \in C_{\bm{z}} \ldotp 
         (\mathit{abt}_{st}((\bm{c}_i(\bm{x}) \leq 0) \land (I(\bm{a}, \bm{x}) \leq 0)) \\
        &\land \mathit{abt}_{st}(\bm{x}' = \bm{f}_i(\bm{x}))) {\implies} 
            \mathit{abt}(I(\bm{a}, \bm{x}') \leq 0) \} \, , 
\end{aligned}
\end{equation*}
for $i = 1, \dots, k$. 

It can be proved that $R_a(i) \subseteq R_I(i)$. 
Note that $R_a(i)$ no longer contain non-polynomial functions 
and can be underapproximated using the techniques presented in previous sections. 
If a valid parameter assignment $\bm{a}_0$ was found in $R_a(i)$, 
an invariant can be obtained by substituting it back to the template. 
It should be noted though that weak completeness 
(Theorem \ref{thm:weakComplete}) does not hold due to information lost in the abstraction process. 

\begin{example}
Consider synthesize invariant clusters for the program given in Code~\ref{code:nonPolynomial} 
using the template $I(a, x, y) = x + y + a \leq 0$. 
\begin{listing}[!ht]
    \begin{minted}[mathescape, escapeinside=||]{c}
        // Precondition: {$x + y + \sin(\pi x) \leq 0$}
        while (x <= 0) {
            y = y + sin(|$\pi$|x);
            x = x + 2;
            y = y - 2;
            y = y - sin(|$\pi$|x);
        }
        // Postcondition: {$x + y + \sin(\pi x) \leq 2$}
    \end{minted}
    \caption{A non-polynomial program}
    \label{code:nonPolynomial}
\end{listing}

The above abstraction process can be applied 
to build the following abstracted valid sets: 
\begin{equation*}
\begin{aligned}
    R_a(0) = \vphantom{1} & \{ a \in C_{a} \mid \forall \bm{z} \in C_{\bm{z}} \ldotp ((z_1 + z_2 + z_3\leq 0) 
        \land (z_3 \geq -1)) \implies (z_1 + z_2 + a \leq 0) \} \, ,  \\
    R_a(1) = \vphantom{1} & \{ a \in C_{a} \mid \forall \bm{z} \in C_{\bm{z}} \ldotp 
        \forall \bm{z'} \in C_{\bm{z}} \ldotp ((z_1 \leq 0) \land (z'_1 = z_1 + 2) \land (z_1 + z_2 + a \leq 0) \\
        &\land (z'_2 = z_2 - 2 + z_3 - z'_3) \land (z_3 = z_3')) \implies (z'_1 + z'_2 + a \leq 0) \}  \, ,  \\
    R_a(2) = \vphantom{1} & \{ a \in C_{a} \mid \forall \bm{z} \in C_{\bm{z}} \ldotp 
        ((z_1 \geq 0) \land (z_1 + z_2 + a \leq 0)) \land (z_3 \leq 1) \implies (z_1 + z_2 + z_3 \leq 2) \} \, .
\end{aligned}
\end{equation*}
where $z_1 = x$, $z_2 = y$ and $z_3 = \sin(x)$. 
Note that the additional formulas $z_3 \geq -1$. $z_3 \leq 1$ and $z_3 = z'_3$ 
introduced by the strengthening procedure.

The abstracted valid set $R_a(i)$ can then be underapproximated 
by the techniques presented in previous sections. 
A valid parameter assignment $a = -1$ can be extracted from it, 
which gives us the invariant ${\{(x, y) \mid x + y - 1 \leq 0\}}$ 
that can be used to verify the safety property. 
\end{example}

 \section{Experiments}
\label{sec:experiments}

The following experiments are performed on a laptop 
with ... 

We used \verb+SumOfSquares.jl+ package\cite{legat2017sos, weisser2019polynomial} of \verb+Julia+\cite{Julia} 
to invoke the SDP solver of \verb+Mosek+\cite{Mosek} 
to solve the resulted semidefinite programming problems. 
For comparison, we also use \verb+Z3+ version \verb+4.8.0+ 
and \verb+Redlog+ version \verb+3258+ 
to do nonlinear real SMT solving and nonlinear real quantifier elimination. 

\subsection{An Illustrative Example} 
Firstly, we use the following simple program to illustrate the ideas 
of using invariant clusters to prove safety property:  
\begin{listing}[!ht]
\begin{minted}[mathescape]{c}
    // Precondition: {$x^2 + y^2 - 1 \leq 0$} 
    while (*) {
        nx = 0.9 * (x - 0.01 * y);
        ny = 0.9 * (y + 0.01 * x);
        (x, y) = (nx, ny);
    }
    // Postcondition: {$x^2 + (y - 2)^2 - 0.25 \geq 0$}
\end{minted}
\caption{An Illustrative Example}
\label{code:illustritive}
\end{listing}

Considering the system dynamics, 
the following template is used to search for ellipsoid-shaped invariants centered at the origin: 
\begin{equation*}
    \textit{Inv}(a, b, x, y) = x^2 + a y^2 + b \, ,
\end{equation*}
where $a$ and $b$ are parameters with range $0 \leq a \leq 10$, $-10 \leq b \leq 0$. 

By applying the methods given in Section~\ref{sec:underapproximation}, 
we obtain the following underapproximation of degree $2$: 
\begin{equation*}
\begin{aligned}
p_1(a, b) &= 0.832285 a^2 + 4.284474 a+ 5.000000 b + 0.503953 \, , \\
    p_2(a, b) &= 0.832672 b - 0.842673 \, , \\
    p_3(a, b) &= 0.424217 a^2 - 0.919546 a - 0.868952 b - 0.827048 \, , 
\end{aligned}
\end{equation*}
where $p_i(a, b)$ underapproximates the $i$-th condition of Definition~\ref{def:inv}. 
The safety property can be proved by finding a valid solution such that $p_i(a, b) \leq 0$ for all $i$. 
There are, obviously, many such valid solutions for the above $p_i(a, b)$, 
each each solution corresponds to an invariant Code~\ref{code:illustritive}. 
For example, one such solution is 
$(a, b) =  (6.503013, -6.765751)$, 
which corresponds to the invariant 
$x^2 + 6.503012 y^2 - 6.765751 \leq 0$. 

Now, suppose the postcondition is changed to $(x - 3)^2 + y^2 - 0.25 \geq 0$. 
The above invariant is no longer capable of proving the new postcondition. 
Instead of repeating the full procedure, 
we only need to adjust $p_3(a, b)$ (the part concerning postconditions). 
Using the same techniques, we can obtain a new underapproximation 
$p'_3(a, b) = -1.708145 b - 0.343592 $. 

The result of $p_1(a, b)$ and $p_2(a, b)$ can be \emph{reused}. 
A new solution such that $p_1(a, b) \leq 0$, $p_2(a, b) \leq 0$ and $p'_3(a, b) \leq 0$ 
can be found as 
$(a, b) = (1.000000, -3.367025)$, 
which corresponds to the invariant 
$x^2 + 1.000000 y^2 - 3.367025 \leq 0$. 

\subsection{Unicycle Model}
In this section, we consider a simple model of a unicycle \cite{sassi2012controller}: 
\begin{equation*}
\begin{aligned}
    \dot{x} &= v \cos(\theta) \, ,  \\
    \dot{y} &= v \sin(\theta) \, ,  \\
    \dot{\theta} &= w \, .
\end{aligned}
\end{equation*}
where $v$ and $w$ are the inputs. 
By using the change of coordinates 
$z_1 = x \cos(w) + y \sin(w)$ and $z_2 = x \sin(w) - y \cos(w)$, 
we can obtain the following polynomial system: 
\begin{equation*}
\begin{aligned}
    \dot{z_1} &= v - z_2 w \, ,   \\
    \dot{z_2} &= z_1 w \, .
\end{aligned}
\end{equation*}

The control program is given in Code~\ref{code:dubin} by 
discretizing the dynamic model with step size $d = 0.01$:
\begin{listing}[!ht]
\begin{minted}[mathescape]{c}
    // Precondition: {$z_1^2 + (z_2 - 1)^2 - 1 \leq 0$} 
    while (*) {
        w = 1.0178 + 1.8721 * z1 - 0.0253 * z2;
        d = 0.01;
        nz1 = z1 + d * (1 - z2 * w);
        nz2 = z2 + d * (z1 * w);
        (z1, z2) = (nz1, nz2);
    }
    // Postcondition: {$z_1^2 + (z_2 - 1)^2 - 4 \leq 0$}
\end{minted}
\caption{The Unicycle Model}
\label{code:dubin}
\end{listing}

Here, the speed $v$ is taken to be $1$ and the control input $w$ is given by \cite{sassi2012controller}.
Note that the results of \cite{sassi2012controller} is done 
in the continuous context with much smaller initial set,
and it remains unclear whether the above \emph{discrete} program 
satisfies the safety property $z_1^2 + (z_2 - 1)^2 - 4 \leq 0$. 

The form of the postcondition indicates that 
we can try the following quadratic template: 
\begin{equation*}
    \textit{Inv}(a, b, c, z_1, z_2) = z_1^2 + a z_2^2 + b z_2 + c \, .
\end{equation*}

With this template, the problem of synthesizing invariants reduces to 
finding valid assignments of $a$, $b$, $c$ 
such that the following first-order formulas are satisfied: 
\begin{equation*}
\begin{aligned}
    C_1 := \forall z_1 \ldotp \forall z_2 \ldotp 
                    & (z_1^2 + (z_2 - 1)^2 - 1 \leq 0) \implies 
                      (z_1^2 + a z_2^2 + b z_2 + c \leq 0) \, ,  \\
    C_2 := \forall z_1 \ldotp \forall z_2 \ldotp 
                    & ((z_1 \leq 0) \land (z_1^2 + a z_2^2 + b z_2 + c \leq 0)) \implies 
                      (a z_1^2 d^2 w^2 + z_2^2 d^2 w^2 + 2 a z_1 z_2 d w + \phantom{1} \\
                    & b z_1 d w - 2 z_1 z_2 d w - 2 z_2 d^2 w + a z_2^2 + 
                      b z_2 + z_1^2 + 2 z_1 d + d^2 + c \leq 0) \, ,  \\
    C_3 := \forall z_1 \ldotp \forall z_2 \ldotp 
                    & ((z_1 \geq 0) \land (z_1^2 + a z_2^2 + b z_2 + c \leq 0)) \implies 
                      (z_1^2 + (z_2 -1)^2 - 4 \leq 0) \, , 
\end{aligned}
\end{equation*}
where $d = 0.01$ and $w = 1.0178 + 1.8721 z_1 - 0.0253 z_2$.

Using the method presented in this paper 
and set the search range of parameters to be $[-5.0, 5.0]$,
the $p_i$ obtained by solving the degree $2$ sum-of-squares relaxation problem is:
\begin{equation*}
\begin{aligned}
p_1(a, b, c) = {} & 0.090529 a^2 + 0.101026 a b + 0.025383 b^2 + 3.868348 a + 
                        1.919025 b   + 1.000000 c + 0.035542 \, ,  \\
    p_2(a, b, c) = {} & 0.000842 a^2 - 0.000807 a b - 0.000331 a c + 0.000619 b^2 + 
                        0.000074 b c + 0.002646 c^2 + \vphantom{1} \\ 
                      & 0.004036 a   - 0.002039 b - 0.000686 c   - 0.004953 \, ,   \\
    p_3(a, b, c) = {} & 1.027814 a^2 + 1.703057 a b + 1.431318 b^2 - 6.012174 a - 
                        7.317851 b   - 1.125513 c - 7.686978 \, .
\end{aligned}
\end{equation*}

A valid assignment 
\begin{equation*}
(a_0, b_0, c_0) = (1.000000, -2.000000, -2.165579)
\end{equation*}
can be extracted using numeric solvers,
which gives:
\begin{equation*}
\mathit{Inv} = z_1^2 + 1.000000 z_2^2 - 2.000000 z_2 - 2.165579 \leq 0 \, .
\end{equation*} 

Subsequent symbolic checks performed by SMT solver \texttt{Z3} confirmed that 
the above $\mathit{Inv}$ is indeed an invariant.

\subsection{Comparison with CODE2INV} 
In this subsection we compare our algorithm against existing works on invariant synthesis. 
Our main focus is on nonlinear invariant synthesis, which has always been a grand challenge. 
Some existing works do not support nonlinear invariant synthesis (such as \cite{CSS03}), 
and some can only synthesize invariants of equality form \cite{joel2018polynomial}. 
Works based on linear recurrence solving (such as \cite{kincaid2017nonlinear}) 
failed on most test cases here due to no non-trivial closed form solution can be found. 
Furthermore, algorithms based on quantified SMT solving or quantifier elimination cannot terminate 
in a reasonable time (over 24 hours) even for the most simple program. 
Implements of some works 
such as \cite{chatterjee2020polynomial} are also not publicly available. 
Therefore, we mainly compare our methods to 
\verb+CODE2INV+ \cite{si2018learning, si2020code2inv}, 
a state-of-the-art invariant synthesis tool 
based on neural network learning. 

We first test our algorithm on the nonlinear programs 
in the benchmark provided by \verb+CODE2INV+ \cite{si2018learning, si2020code2inv}. 
They can be found on \url{https://github.com/PL-ML/code2inv}. 

It should be noted that our algorithm 
only considers real variables and invariants of inequality form. 
Some of the benchmarks are modified slightly so that our algorithm may apply. 

More specifically, our modifications include:
\begin{itemize}
    \item For those test cases that require invariants of equality form, 
          we break the verification task into two tasks 
          and try to prove the ``less than'' part and ``greater than'' part respectively.
    \item For those test cases that require the properties of integer arithmetic, 
          we try to relax them a little bit so that the safety property still holds 
          even in the real context. 
\end{itemize}

Regarding the choice of templates in those test cases, 
first we try a template that include all monomials in postconditions. 
If that does not work, we include monomials that appeared in the programs one by one. 

The experiment results are summarized in Table \ref{tab:code2inv}. 
Time out is set to $60$ minutes. 

\begin{table}[H]
\caption{CODE2INV Nonlinear Benchmark}
\label{tab:code2inv}
\begin{minipage}{\columnwidth}
\begin{center}
\begin{tabular}{lll}
\toprule
     & Time (ours) & Time (\verb+CODE2INV+) \\
\hline
\verb+nl-1+ & 1m40.658s & 1.763s \\
\verb+nl-2+ & 1m43.269s & 1m26.493s\\
\verb+nl-3+ & 2m2.762s & Time Out \\
\verb+nl-4+ & 29m29.777s & Time Out \\
\verb+nl-5+ & Time Out & 3m50.342s \\
\verb+nl-6+ & Time Out & Time Out \\
\verb+nl-7+ & Time Out & 34.948s \\
\end{tabular}
\end{center}
\end{minipage}
\end{table}

These experiments results show that our algorithm 
is comparable to \verb+CODE2INV+ on its nonlinear benchmarks. 
As \verb+CODE2INV+ provides no guarantee of any sense of completeness at all, 
(our algorithm, on the other hand, 
has a weak completeness result stated as Theorem~\ref{thm:weakComplete}), 
it seems to be safe to state that our algorithm 
achieve the goal of getting both theoretical completeness and practical performance. 

Next, we test our methods on a series of control programs 
obtained by discretizing nonlinear dynamical systems. 
This type of programs typically appears 
when simulating or controlling dynamical systems. 
The programs \verb+dubins+ and \verb+dubins_disturbed+ 
have been given in previous subsections, 
and programs \verb+L1+ to \verb+L6+ are given in supplemental text. 
Though the safety of their continuous counterparts 
can be proved relatively easily, 
it remains a question whether they are still safe 
after discretizing (with certain step size $d$).

In all these cases, the range of variable is set to $[-100, 100]$ 
and range of template parameter is set to $[-5, 5]$ during computation. 
Time out for \verb+L1+ to \verb+L6+ is set to $6$ hours, 
and time out for \verb+dubins+ and \verb+dubins_disturbed+ is set to $24$ hours. 

The template is chosen to include all monomials appeared in the postcondition. 
If that does not work, monomials appeared in the programs are added one by one. 
Table~\ref{tab:benchmark} summarises the experiment results. 

\begin{table}[H]
\caption{Control Programs of Nonlinear Systems}
\label{tab:benchmark}
\begin{minipage}{\columnwidth}
\begin{center}
\begin{tabular}{lll}
\toprule
   & Time (ours) & Time (\verb+CODE2INV+) \\
\hline
\verb+L1+ & 3m48.467s & 43m23.052s \\
\verb+L2+ & 4m19.208s  & Exceptions\footnote{\texttt{CODE2INV} reports exceptions \label{refnote}} \\ 
\verb+L3+ & 4m10.369s  & Exceptions\footref{refnote} \\
\verb+L4+ & 4h37m8.891s  & Time Out \\ 
\verb+L5+ & 5m26.691s & Time Out \\
\verb+L6+ & Time Out & Time Out \\
\verb+dubins+ & 39m13.291s & Time Out \\
\verb+dubins_disturbed+ & 23h45m43.014s & Time Out \\
\end{tabular}
\end{center}
\end{minipage}
\end{table}

It can be seen from these results that our algorithm 
significantly outperforms \verb+CODE2INV+ in these test cases. 
The main reason seems to be that these control programs tend to 
have much more complicated behaviours (such as cubic terms), 
which significantly slows down the computations of \verb+CODE2INV+. 
On the other hand, our algorithm scales much better 
when higher degrees terms are presented.

 \section{Conclusions and Future Work}
\label{sec:conclusion}
In this paper, we presented a novel way 
to synthesize basic semialgebraic invariants using SDP 
based on Lasserre's results in \cite{Lasserre10, Lasserre15}. 
Unlike symbolic methods such as SMT solving or quantifier elimination, 
our approach admits the efficiency brought by SDP solving 
and outperforms them greatly 
(as the decision process they used have double exponential complexity \cite{DH88}); 
on the other hand, we also proved a weak completeness result stating that 
when some non-degenerate conditions are satisfied, 
our algorithm guarantees to find an invariant, 
the like of which is previously only provided by symbolic methods. 

In future work, we are interested in exploiting more advanced methods 
to deal with general semialgebraic templates. 
Moreover, we plan to extend the techniques presented in this paper 
to invariant synthesis for polynomial dynamical systems and hybrid systems. 
We also would like to investigate the possible use of moments 
in invariant synthesis for stochastic dynamical systems. 
Finally, we would like to further analyse the numeric errors
introduced by solving SDP in our algorithm.

\bibliography{inv}

\end{document}